\documentclass[copyright]{eptcs}
\providecommand{\event}{J.~Dassow, G.~Pighizzini, B.~Truthe (Eds.): 11th International Workshop\\
on Descriptional Complexity of Formal Systems (DCFS 2009)} 
\providecommand{\volume}{3}
\providecommand{\anno}{2009}
\providecommand{\firstpage}{199} 
\providecommand{\eid}{19}
\usepackage{breakurl}        

\input{dcfs.tex}

\usepackage{amsmath, amssymb, amsbsy}
\usepackage{empheq}
\usepackage{graphicx}
\usepackage{subfigure}

\newcommand{\nat}{\mathbb{N}}
\newcommand{\To}{\Rightarrow}
\newcommand{\CF}{{\bf CF}}
\newcommand{\MAT}{{\bf MAT}}
\newcommand{\VEC}{{\bf V}}
\newcommand{\SMAT}{{\bf sMAT}}
\newcommand{\GS}{{\bf GS}}
\newcommand{\suchthat}{:}
\newcommand{\infinite}{\mathit{inf}}
\newcommand{\finite}{\mathit{fin}}
\newcommand{\cb}{\mathit{cb}}
\newcommand{\lW}{\lambda}
\newcommand{\Powerset}[1]{\mathcal{P}({#1})}

\newcommand{\bdisplay}{\begin{description}\footnotesize\item[]}
\newcommand{\edisplay}{\end{description}}

\newcommand{\bquot}[1]{\begin{quotation}\small\noindent
  \textbf{#1}\hspace{\labelsep}\ignorespaces}
\newcommand{\equot}{\unskip\end{quotation}}

\begin{document}
\title{Capacity Bounded Grammars and Petri Nets}
\def\titlerunning{Capacity Bounded Grammars and Petri Nets}
\def\authorrunning{R.~Stiebe, S.~Turaev}

\author{Ralf Stiebe
\institute{Fakult{\"a}t f{\"u}r Informatik\\
  Otto-von-Guericke-Universit{\"a}t Magdeburg\\
  PF 4120 --
  D-39106 Magdeburg --
  Germany}
 \email{stiebe@iws.cs.uni-magdeburg.de}
\and
Sherzod Turaev
\institute{Universitat Rovira i Virgili\\ Facultat de Lletres -- GRLMC\\
E-43005 Tarragona -- Spain}
\email{sherzod.turaev@urv.cat}
}
\maketitle

\begin{abstract}
  A capacity bounded grammar is a grammar whose derivations are restricted by assigning a bound to the number of every nonterminal symbol in the   sentential forms. In the paper the generative power and closure properties of capacity bounded grammars and their Petri net controlled counterparts are investigated.
\end{abstract}

\section{Introduction}
\label{sec:introduction}

The close relationship between Petri nets and language theory has been extensively studied for a long time \cite{cre:man,das:pau}. Results from the theory of Petri nets have been applied successfully to provide elegant solutions to complicated problems from language theory \cite{esp,hau:jan}.

A context-free grammar can be associated with a context-free (communica-tion-free) Petri net, whose places and  transitions, correspond to the nonterminals and the rules of the grammar, respectively, and whose arcs and weights reflect the change in the number of nonterminals when applying a rule.
In some recent papers, context-free Petri nets enriched by additional components have been used to define regulation mechanisms for the defining grammar \cite{das:tur,tur}.  Our paper continues the research in this direction by restricting the (context-free or extended) Petri nets with place capacity.

Quite obviously, a context-free Petri net with place capacity regulates the defining grammar by permitting only those derivations where the number of each nonterminal in each sentential form is bounded by its capacity. A similar mechanism was discussed in \cite{gin:spa1} where the total number of nonterminals in each sentential form is bounded by a fixed integer. There it was shown that grammars regulated in this way generate the family of context-free languages of finite index, even if arbitrary nonterminal strings are allowed as left-hand sides. The main result of this paper is that, somewhat surprisingly, grammars with capacity bounds have a greater generative power.

This paper is organized as follows. Section~\ref{sec:def} contains some necessary definitions and notations from language and Petri net theory. The concepts of grammars with capacities and grammars controlled by Petri nets with place capacities are introduced in section~\ref{sec:capacities}. The generative power and closure properties of capacity-bounded grammars are investigated in sections \ref{sec:power-gs} and \ref{sec:nb-cfg}. Results on grammars controlled by Petri nets with place capacities are given in section~\ref{sec:PNC}.

\section{Preliminaries}
\label{sec:def}

Throughout the paper, we assume that the reader is familiar with basic concepts of formal language theory and Petri net theory; for details we refer to \cite{das:pau,han,rei:roz}.

The set of natural numbers is denoted by $\nat$, the power set of a set S by $\Powerset{S}$. We use the symbol $\subseteq$ for inclusion and $\subset$ for proper inclusion. The \emph{length} of a string $w \in X^*$ is denoted by $|w|$, the number of occurrences of a symbol $a$ in $w$ by $|w|_a$ and the number of occurrences of symbols from $Y\subseteq X$ in $w$ 
by~$|w|_Y$. The \emph{empty} string is denoted by~$\lW$.

A \emph{phrase structure grammar} (due to Ginsburg and Spanier \cite{gin:spa1}) is a 
quadruple $G=(V, \Sigma, S, R)$ where~$V$ and $\Sigma$  are two finite disjoint alphabets
of \emph{nonterminal} and \emph{terminal} symbols, respectively, $S\in V$ is the 
\emph{start symbol} and \hbox{$R\subseteq V^+\times (V\cup \Sigma)^*$} is a finite 
set of \emph{rules}.

A string $x\in (V\cup \Sigma)^*$ \emph{directly derives} a string $y\in (V\cup \Sigma)^*$ in $G$, written as $x\To y$, if and only if there is a rule $u\to v\in R$ such that $x=x_1ux_2$ and $y=x_1vx_2$ for some $x_1, x_2\in (V\cup \Sigma)^*$. The reflexive and transitive closure of the relation $\To$ is denoted by $\To^*$. A derivation using the sequence of rules $\pi=r_1r_2\cdots r_k$, $r_i\in R$, $1\leq i\leq k$, is denoted by $\xRightarrow{\pi}$ or $\xRightarrow{r_1r_2\cdots r_k}$. The \emph{language} generated by $G$, denoted by $L(G)$, is defined by $L(G)=\{w\in \Sigma^*\suchthat S\To^* w\}.$ A phrase structure grammar $G=(V, \Sigma, S, R)$ is called \emph{context-free} if each rule $u\to v\in R$ has $u\in V$.
The family of context-free languages is denoted by $\mathbf{CF}$.

A \emph{matrix grammar} is a quadruple $G=(V, \Sigma, S, M)$ where $V, \Sigma, S$ are defined as for a context-free grammar, $M$ is a finite set of \emph{matrices} which are finite strings (or finite sequences) over a set of context-free rules. The language generated by the grammar $G$ consists of all strings $w\in \Sigma^*$ such that there is a derivation $S\xRightarrow{r_1r_2\cdots r_n}w$ where $r_1r_2\cdots r_n$ is a concatenation of some matrices $m_{i_1}, m_{i_2}, \ldots, m_{i_k}\in M$, $k\geq 1$. The family of languages generated by matrix grammars without erasing rules (with erasing rules, respectively)  is denoted by $\mathbf{MAT}$ (by $\mathbf{MAT}^{\lW}$, respectively).

A \emph{vector grammar} is defined like a matrix grammar, but the derivation sequence $r_1r_2\cdots r_n$ has to be a shuffle of some matrices $m_{i_1}, m_{i_2}, \ldots, m_{i_k}\in M$, $k\geq 1$.
A \emph{semi-matrix grammar} is defined like a matrix grammar, but the derivation sequence $r_1r_2\cdots r_n$ has to be the semi-shuffle of some matrices $m_{i_1}, m_{i_2}, \ldots, m_{i_k}\in M$, $k\geq 1$, i.\,e., from the shuffle of sequences from $\bigcup_{i=1}^t m_i^*$ where
$$M=\{m_1,\ldots,m_t\}.$$
The language families generated by vector and semi-matrix grammars are denoted by $\VEC^{[\lW]}$ and $\SMAT^{[\lW]}$.

\medskip

A \emph{Petri net} (PN) is a construct $N = (P, T, F, \phi)$ where $P$ and $T$ are disjoint finite sets of \emph{places} and \emph{transitions}, respectively,  $F \subseteq (P\times T) \cup (T\times P)$ is the set of \emph{directed arcs}, 
$$\varphi: (P\times T) \cup (T\times P) \rightarrow \{0, 1, 2, \dots\}$$ 
is a \emph{weight function}, where $\varphi(x,y)=0$ for all $(x,y)\in ((P\times T) \cup (T\times P))-F$. A mapping 
$$\mu: P \rightarrow \{0,1,2, \ldots\}$$ 
is called a \emph{marking}. For each place $p\in P$, $\mu(p)$ gives the number of \emph{tokens} in $p$. $^{\bullet}x=\{y\suchthat \, (y,x)\in F\}$ and $x^{\bullet}=\{y\suchthat \, (x,y)\in F\}$ are called the sets of \emph{input} and \emph{output} elements of $x\in P\cup T$, respectively.

A sequence of places and transitions $\rho=x_1x_2\cdots x_n$ is called a \emph{path} if and only if no place or transition except $x_1$ and $x_n$ appears more than once, and $x_{i+1}\in x^\bullet_{i}$ for all $1\leq i\leq n-1$.  We denote by $P_\rho, T_\rho, F_\rho$ the sets of places, transitions and arcs of $\rho$. Two paths $\rho_1$, $\rho_2$ are called \emph{disjoint} if $P_{\rho_1}\cap P_{\rho_2}=\emptyset$ and $T_{\rho_1}\cap T_{\rho_2}=\emptyset$.
A path $\rho=t_{1}p_{1}t_{2}p_{2}\cdots p_{k-1}t_{k}$ ($\rho=p_{1}t_{1}p_{1}t_{2}\cdots t_{k}p_{1}$) is called a \emph{chain} (\emph{cycle}).

A transition $t \in T$ is \emph{enabled} by marking $\mu$ iff $\mu(p)\geq \phi(p,t)$ for all $p\in P$. In this case $t$ can \emph{occur}. Its occurrence transforms the marking $\mu$ into the marking $\mu'$ defined for each place $p \in P$ by $\mu'(p)=\mu(p)-\phi(p,t)+\phi(t,p)$. This transformation is denoted by $\mu\xrightarrow{t}\mu'$. A finite sequence $t_1t_2\cdots t_k$ of transitions is called \emph{an occurrence sequence} enabled at a marking $\mu$ if there are markings $\mu_1, \mu_2, \ldots, \mu_k$ such that $\mu \xrightarrow{t_1} \mu_1 \xrightarrow{t_2} \ldots \xrightarrow{t_k} \mu_k$.  For each $1\leq i\leq k$, marking $\mu_i$ is called \emph{reachable} from marking $\mu$. $\mathcal{R}(N, \mu)$ denotes the set of all reachable markings from a marking $\mu$.

A \emph{marked} Petri net is a system $N=(P, T, F, \phi, \iota)$ where $(P, T, F, \phi)$ is a Petri net, $\iota$ is the \emph{initial marking}. Let $M$ be a set of markings, which will be called \emph{final} markings. An occurrence sequence $\nu$ of transitions is called \emph{successful} for $M$ if it is enabled at the initial marking $\iota$ and finished at a final marking $\tau$ of $M$.

A Petri net $N$ is said to be $k$-\emph{bounded} if the number of tokens in each place does not exceed a finite number $k$ for any marking reachable from the initial marking $\iota$, i.\,e., $\mu(p)\leq k$ for all $p\in P$ and for all $\mu\in \mathcal{R}(N, \iota)$. A Petri net is called \emph{bounded} if it is $k$-bounded for some $k\geq 1$.

A Petri net with \emph{place capacity} is a system $N=(P, T, F, \phi, \iota,\kappa)$ where $(P, T, F, \phi,\iota)$ is a marked Petri net and $\kappa:P \to \nat$ is a function assigning to each place a number of maximal admissible tokens.
A  marking $\mu$ of $N$ is valid if $\mu(p)\leq \kappa(p)$, for each place $p\in P$. A transition $t \in T$ is \emph{enabled} by a marking $\mu$ if additionally the successor marking is valid.


A \emph{cf Petri net} with respect to a context-free grammar $G=(V,\Sigma, S, R)$ is a system 
$$N=(P, T, F, \phi, \beta, \gamma, \iota)$$
where
\begin{itemize}
\item labeling functions $\beta:P\rightarrow V$ and $\gamma:T\rightarrow R$ are bijections;
\item $(p,t)\in F$ iff $\gamma(t)=A\rightarrow \alpha$ and $\beta(p)=A$ and the weight of the arc $(p,t)$ is 1;
\item $(t,p)\in F$ iff $\gamma(t)=A\rightarrow \alpha$, $\beta(p)=x$ where $|\alpha|_x>0$ and the weight of the arc $(t,p)$ is $|\alpha|_x$;
\item the initial marking $\iota$ is defined by $\iota(\beta^{-1}(S))= 1$  and $\iota(p) = 0$ for all $p\in P-\beta^{-1}(S)$.
\end{itemize}

Further we recall the definitions of extended cf Petri nets, and grammars controlled by these Petri nets (for details, see \cite{das:tur,tur}).

Let $G=(V, \Sigma, S, R)$ be a context-free grammar with its corresponding cf Petri net 
$$N=(P, T, F, \phi, \beta, \gamma, \iota).$$
Let $T_1, T_2, \ldots, T_n$ be a partition of $T$.

1. Let $\Pi=\{\rho_1, \rho_2, \ldots, \rho_n\}$ be the set of disjoint chains such that $T_{\rho_i}=T_i$, $1\leq i\leq n$, and 
$$\bigcup_{\rho\in\Pi}P_\rho\cap P=\emptyset.$$
An \emph{$h$-Petri net} is a system 
$N_h=(P\cup Q, T, F\cup E, \varphi, \zeta, \gamma, \mu_0, \tau)$
where \hbox{$Q=\bigcup_{\rho\in\Pi}P_\rho$} and  $E=\bigcup_{\rho\in\Pi}F_\rho$; the weight function $\varphi$ is defined by $\varphi(x,y)=\phi(x,y)$ if $(x,y)\in F$ and $\varphi(x,y)=1$ if $(x,y)\in E$;
the labeling function $\zeta:P\cup Q\rightarrow V\cup\{\lambda\}$ is defined by $\zeta(p)=\beta(p)$ if $p\in P$ and $\zeta(p)=\lambda$ if $p\in Q$; the initial marking $\mu_0$ is defined by $\mu_0(p)=\iota(p)$ if $p\in P$ and $\mu_0(p)=0$ if $p\in Q$; $\tau$ is the final marking where $\tau(p)=0$ for all $p\in P\cup Q$.

2. Let $\Pi=\{\rho_1, \rho_2, \ldots, \rho_n\}$ be the set of disjoint cycles such that $T_{\rho_i}=T_i$, $1\leq i\leq n$, and 
$$\bigcup_{\rho\in\Pi}P_\rho\cap P=\emptyset.$$
A \emph{$c$-Petri net} is a system 
$N_c=(P\cup Q, T, F\cup E, \varphi, \zeta, \gamma, \mu_0, \tau)$
where $Q=\bigcup_{\rho\in\Pi}P_\rho$ and $E=\bigcup_{\rho\in\Pi}F_\rho$; the weight 
function $\varphi$ is defined by $\varphi(x,y)=\phi(x,y)$ if \hbox{$(x,y)\in F$} and 
$\varphi(x,y)=1$ if $(x,y)\in E$; the labeling
function $\zeta:P\cup Q\rightarrow V\cup\{\lambda\}$ is defined by $\zeta(p)=\beta(p)$ 
if $p\in P$ and $\zeta(p)=\lambda$ if $p\in Q$; the initial marking $\mu_0$ is defined
by $\mu_0(p)=\iota(p)$ if $p\in P$, and $\mu_0(p_{i,1})=1$, $\mu_0(p_{i,j})=0$ 
where $p_{i,j}\in P_i$, $1\leq i\leq n$, $2\leq j\leq k_i$; $\tau$ is the final marking 
where $\tau(p)=0$ if $p\in P$, and $\tau(p_{i,1})=1$, $\tau(p_{i,j})=0$ where 
$p_{i,j}\in P_i$, $1\leq i\leq n$, $2\leq j\leq k_i$.

3. Let $\Pi=\{\rho_1, \rho_2, \ldots, \rho_n\}$ be the set of cycles such that $T_{\rho_i}\!=T_i$, $1\leq i\leq n$, \hbox{$P_{1}\cap P_{2}\cap \cdots \cap P_{n}\!=\!\{p_0\}$} and $$\bigcup_{\rho\in\Pi}P_\rho\cap P=\emptyset.$$
An \emph{$s$-Petri net} is a system $N_s=(P\cup Q, T, F\cup E, \varphi, \zeta, \gamma, \mu_0, \tau)$ where $Q=\bigcup_{\rho\in\Pi}P_\rho,   E=\bigcup_{\rho\in\Pi}F_\rho$; the weight function $\varphi$ is defined by $\varphi(x,y)=\phi(x,y)$ if $(x,y)\in F$ and $\varphi(x,y)=1$ if $(x,y)\in E$; the labeling function $\zeta:P\cup Q\rightarrow V\cup\{\lambda\}$ is defined by $\zeta(p)=\beta(p)$ if $p\in P$ and $\zeta(p)=\lambda$ if $p\in Q$; $\mu_0$ is the initial marking where $\mu_0(p_0)=1$ and $\mu_0(p)=\iota(p)$ if $p\in (P\cup Q)-\{p_0\}$; $\tau$ is the final marking where $\tau(p_0)=1$ and $\tau(p)=0$ if $p\in (P\cup Q)-\{p_0\}$.

\begin{example}
Figure \ref{fig:xPNs} depicts extended cf Petri nets which are constructed with respect to the context-free grammar $G'=(\{S, A, B\}, \Sigma, S, R)$ where $R$ consists of $r_0: S\to AB$, $r_1: A\to \lambda$, $ r_3: A\rightarrow aA$, $r_5: A\to bA$, $r_2: B\to \lambda$, $r_4: B\to aB$, $r_6: B\to bB$.%
\hfill$\diamond$
\end{example}

\begin{figure}[ht]
  \begin{center}
    \subfigure[an $h$-Petri net]{\label{fig:hPN}\includegraphics[width=0.3\textwidth]{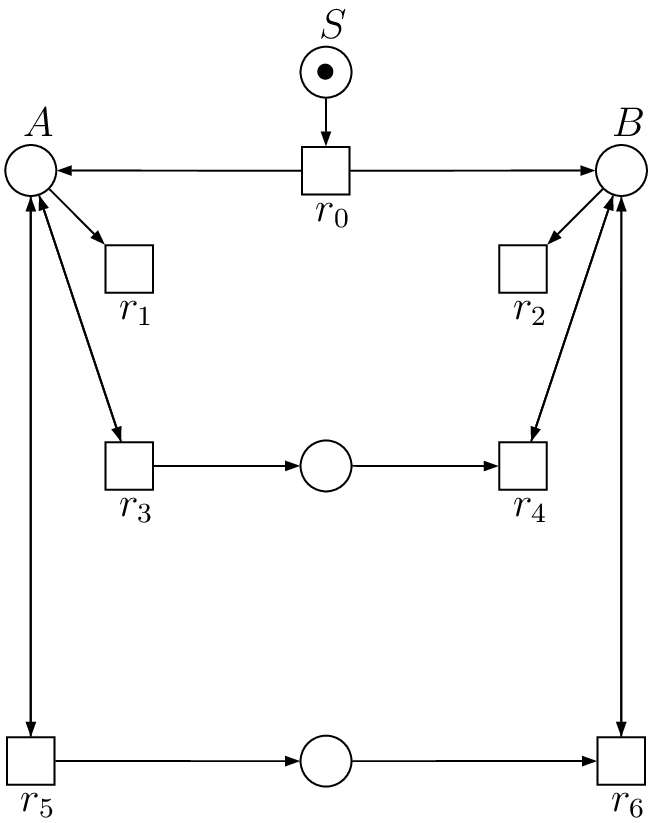}}
    \subfigure[a $c$-Petri net]{\label{fig:cPN}\includegraphics[width=0.3\textwidth]{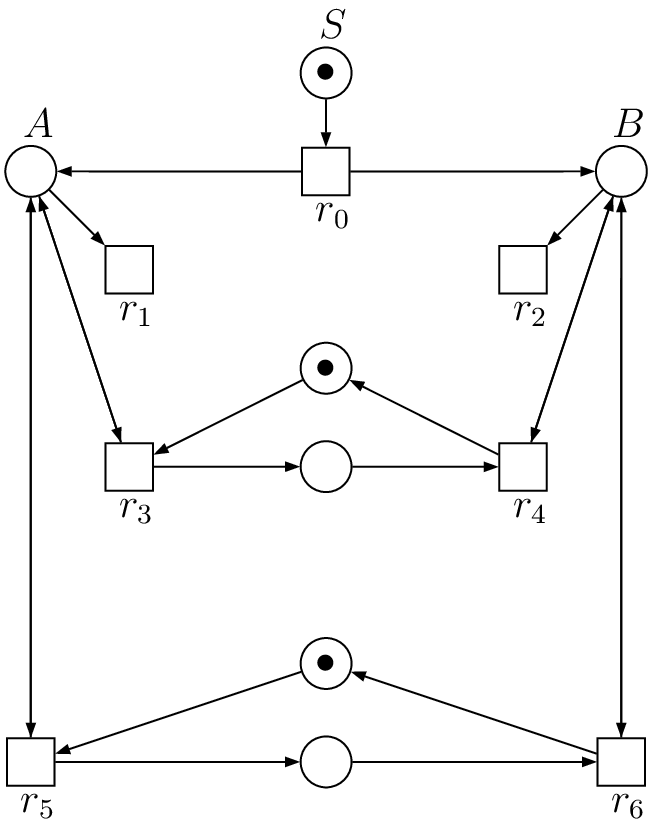}}
    \subfigure[an $s$-Petri net]{\label{fig:sPN}\includegraphics[width=0.3\textwidth]{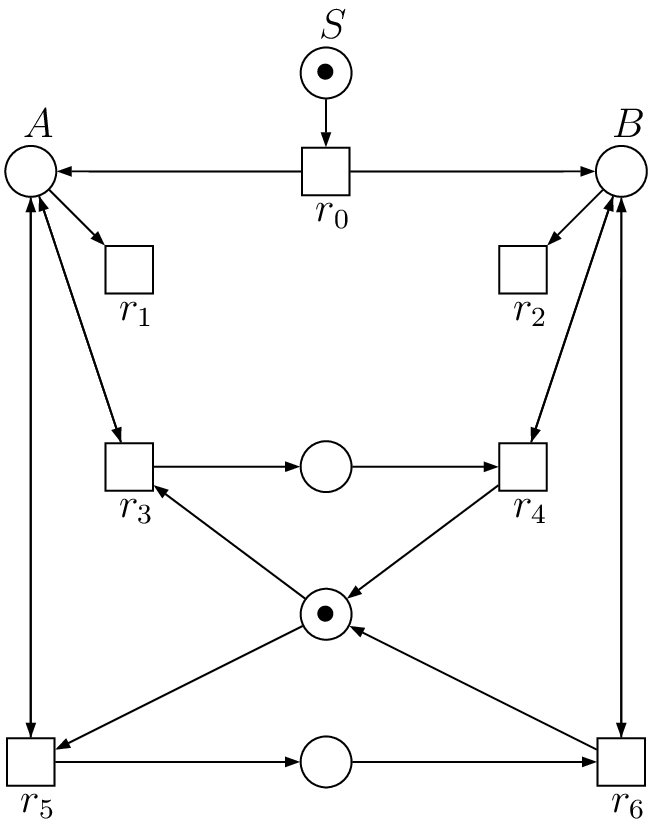}}
  \end{center}
  \caption{Extended Petri nets.}\label{fig:xPNs}
\end{figure}

A \emph{$z$-PN controlled grammar} is a system $G=(V, \Sigma, S, R, N_z)$ where
\hbox{$G'=(V, \Sigma, S, R)$} is a context-free grammar and $N_z$ is  $z$-Petri net with 
respect to the context-free grammar $G'$ where $z\in\{h, c, s\}$. The \emph{language} 
generated by a $z$-Petri net controlled grammar $G$ consists of all strings $w\in \Sigma^*$
such that there is a derivation $S\xRightarrow{r_1r_2\cdots r_k}w\in \Sigma^*$  and a 
successful occurrence sequence of transitions $\nu=t_1t_2\cdots t_k$ of~$N_z$ such 
that $r_1r_2\cdots r_k=\gamma(t_1t_2\cdots t_k)$.

\section{Grammars and Petri nets with capacities}
\label{sec:capacities}
We will now introduce grammars with capacities and show some relations to similar concepts known from the literature.

A \emph{capacity-bounded} grammar is a quintuple $G=(V,\Sigma,S,R,\kappa)$ where
$G'=(V,\Sigma,S,R)$ is a grammar and $\kappa: V \to \nat$ is a capacity function. 
The language of $G$ contains all words $w \in L(G')$ that have a derivation $S\To^* w$ such that $|\beta|_A\leq \kappa(A)$ for all $A\in V$ and each sentential form $\beta$ of the derivation. The families of languages generated by arbitrary capacity-bounded grammars (due to Ginsburg and Spanier) and by context-free capacity-bounded grammars are denoted by $\mathbf{GS}_{\cb}$ and $\mathbf{CF}_{\cb}$, respectively. The capacity function mapping each nonterminal to $1$ is denoted by $\mathbf{1}$.

Capacity bounded grammars are closely related to nonterminal-bounded, deri-vation-bounded and finite index grammars.
A grammar $G=(V,\Sigma,S,R)$ is \emph{nonterminal bounded} if $|\beta|_V\leq k$ for some fixed $k \in \nat$ and all sentential forms $\beta$ derivable in $G$. The \emph{index} of a derivation in $G$ is the maximal number of nonterminal symbols in its sentential forms. $G$ is of \emph{finite index} if every word in $L(G)$ has a derivation of index at most $k$ for some fixed $k\in \nat$. The family of context-free languages of finite index is denoted by
$\mathbf{CF}_{\finite}$.
A \emph{derivation-bounded} grammar is a quintuple $G=(V,\Sigma,S,R,k)$ where $G'=(V,\Sigma,S,R)$ is a grammar and $k \in \nat$ is a bound on the number of allowed nonterminals. The language of $G$ contains all words $w \in L(G')$ that have a derivation $S\To^* w$ such that $|\beta|_V\leq k$, for each sentential form $\beta$ of the derivation. It is well-known that the family of derivation bounded languages is equal to $\mathbf{CF}_{\finite}$, even if arbitrary grammars due to Ginsburg and Spanier are permitted \cite{gin:spa2}.

\begin{example}\label{exa:NBLnotCF1}
Let $G=(\{S, A, B, C, D, E, F\}, \{a, b, c\}, S, R,\mathbf{1})$ be the capacity-bounded grammar where $R$ consists of the rules:
$$
\begin{array}{llll}
r_1: S\to ABCD, & r_2: AB\to aEFb, & r_3: CD\to cAD, &r_4: EF\to EC,\\
r_5: EF\to FC,  & r_6: AD\to FD,   & r_7: AD\to ED,  & r_8: EC\to AB,\\
r_9: FD\to CD,  & r_{10}: FC\to AF,& r_{11}: AF\to \lW, & r_{12}: ED\to \lW.
\end{array}
$$

The possible derivations are exactly those of the form
$$
\begin{array}{ll}
S &\xRightarrow{r_1}ABCD\xRightarrow{(r_2r_3r_4r_6r_8r_9)^n}a^nABb^nc^nCD
\xRightarrow{r_2r_3}a^{n+1}EFb^{n+1}c^{n+1}AD \\
& \xRightarrow{r_5r_7}a^{n+1}FCb^{n+1}c^{n+1}ED\xRightarrow{r_{10}r_{11}r_{12}}a^nb^nc^n
\end{array}
$$
(in the last phase, the sequences $r_{10}r_{12}r_{11}$ and $r_{12}r_{10}r_{11}$ could also be applied with the same result).
Therefore, $L(G)=\{a^nb^nc^n\suchthat n\geq 1\}$.\hfill$\diamond$
\end{example}

\begin{example}\label{exa:NBLnotCF2}
Let $G=(\{S,A,B,C\},\{a,b,c\},S,R,\mathbf{1})$ be the context-free capa-city-bounded grammar where $R$ consists of the rules
$r_1: S\to aBbaAb$, $r_2: A\to aBb$, $r_3: B\to C$, $r_4: C\to A$, 
$r_5: A\to BC$, $r_6: A\to c$,
and let $M$ be the regular set $M=\{a^*ccb^*a^*cb^*\}$. The derivations in $G$ generating words from $M$ are exactly those of the form
$$
\begin{array}{ll}
S &\xRightarrow{r_1}aBbaAb\xRightarrow{(r_3r_2r_4r_3r_2r_4)^n}a^nBb^na^nAb^n
\xRightarrow{r_6r_3r_4}a^nAb^na^ncb^n\\
&\xRightarrow{(r_2r_3r_4)^m}a^{n+m}Ab^{n+m}a^ncb^n
\xRightarrow{r_5r_4r_3r_6r_4r_6}a^{n+m}ccb^{n+m}a^ncb^n
\end{array}
$$
(one can also apply $r_3r_6r_4$ in the third phase and $r_5r_4r_6r_3r_4r_6$ in the last phase with the same result).
Hence,
$
L(G)\cap M=\{a^nccb^na^mcb^m\suchthat n\geq m\geq 1\}\not\in \mathbf{CF},
$
implying that $L(G)$ is not context-free.\hfill$\diamond$
\end{example}

The above examples show that capacity-bounded grammars -- in contrast to derivation bounded grammars -- can generate non-context-free languages.
The generative power of capacity-bounded grammars will be studied in detail in the following two sections.

The notions of finite index and bounded capacities can be extended to matrix, vector and semi-matrix grammars. The corresponding language families are denoted by $\MAT^{[\lW]}_{\finite}$, $\VEC^{[\lW]}_{\finite}$, $\SMAT^{[\lW]}_{\finite}$, $\MAT^{[\lW]}_{cb}$, $\VEC^{[\lW]}_{cb}$, $\SMAT^{[\lW]}_{cb}$.

Also control by Petri nets can in  a natural way be extended to Petri nets with place capacities. Since an extended cf Petri net $N_z$, $z\in\{h, c, s\}$, has two kinds of places, i.\,e., places labeled by nonterminal symbols and \emph{control} places, it is interesting to consider two types of place capacities in the Petri net: first, we demand that only the places labeled by nonterminal symbols are with capacities (\emph{weak capacity}), and second, all places of the net are with capacities (\emph{strong capacity}).

A $z$-Petri net $N_z=(P\cup Q, T, F\cup E, \varphi, \zeta, \gamma, \mu_0, \tau)$ is with \emph{weak capacity} if the corresponding cf Petri net $(P, T, F, \phi, \iota)$ is with place capacity, and \emph{strong capacity} if the Petri net $(P\cup Q, T, F\cup E, \varphi, \mu_0)$ is with place capacity.
A grammar controlled by a $z$-Petri net with \emph{weak} (\emph{strong}) \emph{capacity} is a $z$-Petri net controlled grammar $G = (V, \Sigma, S, R, N_z)$ where $N_z$ is with weak (strong) place capacity.
We denote the families of languages generated by grammars (with erasing rules) controlled by $z$-Petri nets with weak and strong place capacities by $\mathbf{wPN}_{cz}$, $\mathbf{sPN}_{cz}$ ($\mathbf{wPN}^\lW_{cz}$, $\mathbf{sPN}^\lW_{cz}$), respectively, where $z\in\{h, c, s\}$.

\section{The power of arbitrary grammars with capacities}
\label{sec:power-gs}

It will be shown in this section that arbitrary grammars (due to Ginsburg and Spanier) with capacity generate exactly the family of matrix languages of finite index. This is in contrast to derivation bounded grammars which generate only context-free languages of finite index.

First we show that we can restrict to grammars with capacities bounded by~$1$. 
Let $\mathbf{CF}_{\cb}^{1}$ and $\mathbf{GS}_{\cb}^{1}$ be the language families generated by context-free and arbitrary grammars with capacity function $\mathbf{1}$.
\begin{lemma}
  $\mathbf{CF}_{\cb}=\mathbf{CF}_{\cb}^{1}$ and $\mathbf{GS}_{\cb}=\mathbf{GS}_{\cb}^{1}$.
\end{lemma}
\begin{proof}
  Let $G=(V,\Sigma,S,R, \kappa)$ be a capacity-bounded phrase structure grammar. We construct the grammar $G'=(V',\Sigma,(S,1),R')$ with capacity function $\mathbf{1}$ and
  \begin{eqnarray*}
    V'&=& \{(A,i) \suchthat A \in V, 1\leq i\leq \kappa(A)\},\\
    R'&=& \{\alpha' \to \beta' \suchthat \alpha' \in h(\alpha), \beta' \in h(\beta), \mbox{ for some } \alpha \to \beta \in R\},
  \end{eqnarray*}
  where $h:(V\cup \Sigma)^* \to (V' \cup \Sigma)^*$ is the finite substitution defined by $h(a)=\{a\}$, for $a \in \Sigma$, and\linebreak
  \hbox{$h(A)=\{(A,i): 1\leq i\leq \kappa(A)\}$}, for $A \in V$.

  It can be shown by induction on the number of derivation steps that $S \!\To^*_{G,\kappa}\! \alpha$ holds iff \hbox{$(S,1) \!\To^*_{G',1}\! \alpha'$}, for some $\alpha' \in h(\alpha)$.
\end{proof}

\begin{lemma}
  \label{lem:GScbSubsetMATfin}
  $\mathbf{GS}_{\cb}\subseteq \mathbf{MAT}_{\finite}$.
\end{lemma}
\begin{proof*}
  Consider some language $L\in \mathbf{GS}_{\cb}$ and let $G=(V,\Sigma,S,R,\mathbf{1})$ be a capacity-bounded phrase structure grammar (due to Ginsburg and Spanier) such that $L=L(G)$. A word $\alpha\in (V\cup \Sigma)^*$ can be uniquely decomposed as
  $$
  \alpha=x_1 \beta_1 x_2 \beta_2 \cdots x_n \beta_n x_{n+1}, x_1,x_{n+1} \in \Sigma^*, x_2,\ldots,x_n \in \Sigma^+, \beta_1,\ldots, \beta_n\in V^+.
  $$
  The subwords $\beta_i$ are referred to as the \emph{maximal nonterminal blocks}  of $\alpha$. Note that the length of a maximal block in any sentential form of a derivation in $G$ is bounded by $|V|$. We will first construct a capacity-bounded grammar $G'$ with $L(G')=L$ such that all words of $L$ can be derived in $G'$ by rewriting a maximal nonterminal block in every step. Let $G'=(V,\Sigma,S,R',\mathbf{1})$ where
  \begin{eqnarray*}
    R'&=& \{\alpha_1 \alpha \alpha_2 \to \alpha_1 \beta \alpha_2 \suchthat
     \alpha \to \beta \in R, \alpha_1,\alpha_2 \in V^*,
     |\alpha_1 \alpha \alpha_2|_A \leq 1, \mbox{ for all } A\in V\}.
  \end{eqnarray*}
  The inclusion $L(G) \subseteq L(G')$ is obvious since $R\subseteq R'$.
  On the other hand, any derivation step in $G'$ can be written as
  $\gamma_1 \underline{\alpha_1 \alpha \alpha_2} \gamma_2 \To_{G'}
  \gamma_1 \underline{\alpha_1 \beta \alpha_2} \gamma_2$, where 
  $\alpha \to \beta \in R$,
  implying that the same step can be performed in $G$ as
  $\gamma_1 \alpha_1 \underline{\alpha} \alpha_2 \gamma_2 \To_{G,1}
  \gamma_1 \alpha_1 \underline{\beta} \alpha_2 \gamma_2.$
  Thus $L(G')\subseteq L(G)$ holds as well.
  Moreover, any derivation step in $G$,
  $\gamma_1 \alpha_1 \underline{\alpha} \alpha_2 \gamma_2 \To_{G,1}
  \gamma_1 \alpha_1 \underline{\beta} \alpha_2 \gamma_2$, 
  $\alpha_1\alpha\alpha_2$ being a maximal nonterminal block,
  can be performed in $G'$  replacing the maximal nonterminal block
  $\alpha_1\alpha\alpha_2$ by $\alpha_1\beta\alpha_2$.

  In the second step we construct a context-free matrix grammar $H$
  which simulates exactly  those derivations in $G'$ that replace a
  maximal nonterminal block in each step. We introduce two alphabets
  \begin{eqnarray*}
    [V]&=&\{[\alpha] \suchthat \alpha \in V^+,
    |\alpha|_A\leq 1, \mbox{ for all } A \in V\}\mbox{ and } \overline{V}=\{\overline{A} \suchthat A \in V\}.
  \end{eqnarray*}
  The symbols of $[V]$ are used to encode each maximal nonterminal block as
  single symbols, while $\overline{V}$ is a disjoint copy of $V$.
  Any word
  $$\alpha=x_1 \beta_1 x_2 \beta_2 \cdots x_n \beta_n x_{n+1},
  x_1,x_{n+1} \in \Sigma^*, x_2,\ldots,x_n \in \Sigma^+,
  \beta_1,\ldots \beta_n\in V^+$$
  such that $|\alpha|_A\leq 1$, for all $A \in V$,
  can be represented by the word
  $[\alpha]=x_1 [\beta_1] x_2 [\beta_2] \cdots x_n [\beta_n] x_{n+1}$,
  where the maximal nonterminal blocks in $\alpha$ are replaced by the
  corresponding symbols from $[V]$.
  The desired matrix grammar is obtained as
  $H=(V_H,\Sigma,S',M)$, with $V_H=[V]\cup V \cup \overline{V} \cup \{S'\}$
  and the set of matrices defined as follows.
  For any rule
  $r=\alpha\to \beta$ in $R'$,
  $M$ contains the matrix $m_r$ consisting of the rules
  \begin{itemize}
  \item $[\alpha] \to [\beta]$
    (note that $\alpha \in [V]$, but $\beta\in ([V]\cup\Sigma)^*$),
  \item $A \to \overline{A}$, for all $A \in V$ such that
    $|\alpha|_A=1$ and $|\beta|_A=0$,
  \item $\overline{A}\to A$, for all $A \in V$ such that
    $|\alpha|_A=0$ and $|\beta|_A=1$.
  \end{itemize}
  (The order of the rules in $m_r$ is arbitrary).
  Additionally, $M$ contains the starting and the terminating matrices
  $$(S'\to [S] S \overline{A_1} \cdots \overline{A_m}) \mbox{ and } (\overline{S} \to \lW,\overline{A_1} \to \lW, \ldots,\overline{A_m} \to \lW),$$
  where $V=\{S,A_1,\ldots,A_m\}$.
  Intuitively, $H$ generates sentential forms of the shape
  $[\beta] \gamma$ where\linebreak $[\beta] \in ([V] \cup \Sigma)^*$ encodes
  a sentential form $\beta$ derivable in $G'$ and
  $\gamma \in (V\cup \overline{V})$ gives a count of the nonterminal symbols
  in $\beta$ as follows: $|\gamma|_A+|\gamma|_{\overline{A}}=1$ and
  $|\gamma|_A=|\beta|_A$.
  Formally, it can be shown by induction that a sentential form over
  $V_H \cup \Sigma$ can be generated after applying $k\geq 1$ matrices
  (except for the terminating)
  iff it has the form $[\beta] \gamma$ where
  \begin{itemize}
  \item $\beta \in (V\cup\Sigma)^*$ can be derived in $G'$ in $k-1$ steps,
  \item $\gamma \in \{S,\overline{S}\} \{A_1,\overline{A}_1\} \cdots\{A_m,\overline{A}_m\}$ and
    $|\gamma|_A=1$ iff $|\beta|_A=1$.\hfill\qed
  \end{itemize}
\end{proof*}
We can also show that the inverse inclusion also holds.
\begin{lemma}\label{MATfinInCScb}
  $\mathbf{MAT}_{\finite}\subseteq \mathbf{GS}_{\cb}$.
\end{lemma}

\section{Capacity-bounded context-free grammars}
\label{sec:nb-cfg}

In this section, we investigate capacity-bounded context-free grammars. It turns out that they are strictly between context-free languages of finite index and matrix languages of finite index. Closure properties of capacity bounded languages with respect to AFL operations are shortly discussed at the end of the section.

As a first result we show that the family of context-free languages with finite index is properly included in $\CF_{\cb}$.
\begin{lemma}
  \label{thm:hierarchyCapacityBounded1}
  $\CF_{\finite} \subset \CF_{\cb}$.
\end{lemma}
\begin{proof}
  Any context-free language generated by a grammar $G$ of index $k$ is also generated by the capacity-bounded grammar $(G,\kappa)$ where $\kappa$ is the capacity function constantly $k$. The properness of the inclusion follows from Example~\ref{exa:NBLnotCF2}.
\end{proof}

An upper bound for $\CF_{\cb}$ is given by the inclusion $\CF_{\cb}\subseteq \GS_{\cb}=\MAT_{\finite}$. We can prove the properness of the inclusion by presenting a language from $\MAT_{\finite} \setminus \CF_{\cb}$.
\begin{lemma}
  $L=\{a^n b^n c^n \suchthat n\geq 1\} \notin \mathbf{CF}_{\cb}$.
\end{lemma}
\begin{proof}
  Consider a capacity-bounded context-free grammar $G=(V,\Sigma,S,R,\mathbf{1})$ such that
  $L \subseteq L(G)$. For $A \in V$, let $G_A=(V,\Sigma,R,A,\mathbf{1})$.
  The following holds obviously for any derivation in $G$ involving $A$: If $\alpha A \beta \To^*_{G} xyz$, where $\alpha,\beta \in (V\cup \Sigma)^*$, $x,y,z \in \Sigma^*$ and $y$ is the yield of $A$, then $y \in L(G_A)$. On the other hand, for all $x,y,z \in \Sigma^*$ such that $y\in L(G_A)$, the relation $xAz \To^*_{G} xyz$ holds.
  The nonterminal set $V$ can be decomposed as $V=V_{\infinite} \cup V_{\finite}$, where
  \begin{eqnarray*}
    V_{\infinite} &=& \{A \in V \suchthat L(G_A) \mbox{ is infinite}\}\mbox{ and } V_{\finite} = \{A \in V \suchthat L(G_A) \mbox{ is finite}\}.
  \end{eqnarray*}
  Let $K$ be a number such that $|w|<K$, for all
  $w \in \bigcup_{A \in V_{\finite}} L(G_A)$.
  Consider the word $w=a^{rK} b^{rK} c^{rK}$, where $r$ is the longest length
  of a right side in a rule of~$R$.
  There is a derivation $S \To^*_{G} w$.
  Consider the last sentential form $\alpha$ in this derivation that contains
  a symbol from $V_{\infinite}$. Let this symbol be $A$.
  All other nonterminals in $\alpha$ are from $V_{\finite}$, and none of them
  generates a subword containing $A$ in the further derivation process.
  We get thus another derivation of $w$ in $G$ by postponing the rewriting
  of $A$ until all other nonterminals have vanished by applying on them
  the derivation sequence of the original derivation.
  This new derivation has the form
  $S\To^*_{G} \alpha \To^*_{G} xAz \To^*_{G} xyz=w.$
  The length of $y$ can be estimated by $|y|\leq rK$, as $A$ is in the
  first step replaced by a word over $(\Sigma \cup V_{\finite})$ of length
  at most $r$.

  By the remarks in the beginning of the proof, any word $xy'z$ with
  $y' \in L(G_A)$ can be derived in $G$. A case analysis shows
  that $xy'z$ is not in $L$, for any $y'\neq y$.
  Hence $L(G) \neq L$.
\end{proof}

The results can be summarized as follows:
\begin{theorem} \label{thm:hierarchyCapacityBounded}
  $\mathbf{CF}_{\finite}\subset \mathbf{CF}_{\cb} \subset \mathbf{GS}_{\cb}=\mathbf{MAT}_{\finite}.$
\end{theorem}

As regards closure properties, we remark that the constructions showing the closure of $\mathbf{CF}$ under homomorphisms, union, concatenation and  Kleene closure can be easily extended to the case of capacity bounded languages.

\begin{theorem}\label{thm:closureCapacityBounded}
  $\mathbf{CF}_{\cb}$ is closed under homomorphisms, union, concatenation and Kleene closure.
\end{theorem}

\begin{proof}
  We give here a proof only for the Kleene closure and leave the other cases to the reader.

  Let $L\in \mathbf{CF}_{\cb}$ and let $G=(V,\Sigma,S,R,\mathbf{1})$ be a context-free grammar such that $L=L(G)$. We construct   $G'=(V\cup\{S'\},\Sigma,S',R\cup \{S'\to SS',S'\to \lW\},\mathbf{1}).$

  Any terminating derivation in $G'$ that applies the rule $S'\to SS'$ $k$ times generates a word\linebreak \hbox{$w=w_1 w_2 \cdots w_k$}, where $w_i$ is the yield of the $i$-th symbol $S$ introduced by $S'\to SS'$. The subderivation from $S$ to $w_i$ only uses rules from $R$. Moreover, any sentential form $\beta_i$ in this  subderivation is the subword of some sentential form $\beta$ in the derivation of $w$ in $G'$. Hence, $|\beta_i|_A \leq |\beta|_A\leq 1$, for all $1\leq i \leq k$ and all $A \in V$. Consequently, $w_i\in L(G)=L$ and $w \in L^*$.

  Conversely, any word $w=w_1 w_2 \cdots w_k$ with $w_i\in L$, for $1\leq i\leq k$, can be obtained in $G'$ by the derivation
  $$
  S'\To SS' \To^* w_1 S' \To w_1SS' \To^* w_1w_2S' \To^* w_1w_2 \cdots w_kS' \To w_1w_2\cdots w_k
  $$
  where the subwords $w_i$ are derived from $S$ as in $G$.
\end{proof}

As regards closure under intersection with regular sets and under inverse homomorphisms, the constructions to show closure of $\mathbf{CF}$ cannot be extended, since they do not keep the capacity bound. We suspect that $\mathbf{CF}_{\cb}$ is not closed under any of these operations.

\section{Control by Petri nets with place capacities}
\label{sec:PNC}

We will first establish the connection between context-free Petri nets with place capacities and capacity-bounded grammars. Later we will investigate the generative power of various extended context-free Petri nets with place capacities.


The proof for the equivalence between context-free grammars and grammars controlled by cf Petri nets can be immediately transferred to context-free grammars and Petri nets with capacities:
\begin{theorem}
  \label{thm:CapacityPetriNetGrammar}
  Grammars controlled by context-free Petri nets with place capacity functions generate the family of capacity-bounded context-free languages.
\end{theorem}

Let us now turn to grammars controlled by extended cf Petri nets with capacities. We will first study the generative power of capacity-bounded matrix and vector grammars, which are closely related to these Petri net grammars.
\begin{theorem}
  \label{thm:matrixGrammarBounds}
  $\MAT_{\finite}=\VEC^{[\lW]}_{\cb}=\MAT^{[\lW]}_{\cb}=\SMAT^{[\lW]}_{\cb}$.
\end{theorem}
\begin{proof}
  We give the proof of $\MAT_{\finite}=\VEC^{\lW}_{\cb}$. The other equalities can be shown in an analogous way. Since $\MAT_{\finite}=\VEC_{\finite}=\VEC^{\lW}_{\finite}$, it suffices to prove $\VEC_{\finite}\subseteq \VEC^{\lW}_{\cb}$ and $\VEC^{\lW}_{\cb}\subseteq \VEC^{\lW}_{\finite}$. The first inclusion is obvious because any vector grammar of finite index $k$ is equivalent to the same vector grammar with capacity function constantly $k$.

  To show $\VEC^{\lW}_{\cb}\subseteq \VEC^{\lW}_{\finite}$, consider a capacity-bounded 
  vector grammar
  $$G=(\{A_0,A_1,\ldots,A_m\},\Sigma,A_0,M,\mathbf{1}).$$
  (The proof that it suffices to consider the capacity function $\mathbf{1}$ is like for usual grammars.) To construct an equivalent vector grammar of finite index, we introduce the new  nonterminal symbols $B_i,B'_i$, $0\leq i\leq m$, $C$, $C'$. For any rule $r: A\to \alpha$, we define the matrix $\mu(r)=(C\to C',s_0,s_1,\ldots,s_m,r,C'\to C)$ such that
  $s_i=B_i \to B'_i$ if $A=A_i$ and $|\alpha|_A=0$, $s_i=B'_i \to B_i$ if $A\neq A_i$ and $|\alpha|_{A_i}=1$, and  $s_i$ is empty, otherwise.

  Now we can construct $G'=(V',\Sigma,S',M')$ where $M'$ contains
  \begin{itemize}
    \item for any matrix $m=(r_1,r_2, \ldots, r_k)$, the matrix $m'=(\mu(r_1), \ldots, \mu(r_k))$,
    \item the start matrix $(S'\to A_0 B_0 B'_1 \cdots B'_m C)$,
    \item the terminating matrix $(C\to \lW, B'_0\to \lW,B'_1\to\lW, \ldots, B'_m\to \lW)$,
  \end{itemize}
  and $V'=V\cup \{B_i,B'_i\suchthat 0\leq i\leq m\} \cup \{S',C,C'\}$.
  The construction of $G'$ allows only derivation sequences where complete   submatrices $\mu(r)$ are applied:   when the sequence $\mu(r)$ has been started, there is no symbol $C$ before   $\mu(r)$ is finished, and no other submatrix can be started.   It is easy to see that $G'$ can generate after applying complete submatrices exactly those words $\beta \gamma C$ such that $\beta \in (V\cup \Sigma)^*$,  
  \hbox{$\gamma\in \{B_0,B'_0\}\{B_1,B'_1\} \cdots \{B_m,B'_m\}$} such that $\beta$ can be derived in $G$ and $|\gamma|_{B_i}=1$ iff $|\beta|_{A_i}=1$. Moreover, $G'$ is of index $2 |V|+1$.
\end{proof}

By constructions similar to those in \cite{tur} and Theorem~\ref{thm:matrixGrammarBounds} we can show with respect to weak capacities:
\begin{theorem}\label{lem:VfinInwPNch}
For $z\in \{h,c,s\}$, $\MAT_{\finite}=\mathbf{wPN}^{[\lW]}_{cz}$.
\end{theorem}

\begin{proof}
We give only the proof for $z=h$. The other equations can be shown using  analogous arguments. By Theorem~\ref{thm:matrixGrammarBounds} it is sufficient to show the inclusions $\VEC_{fin}\subseteq\mathbf{wPN}_{ch}$ and $\mathbf{wPN}^{\lW}_{ch}\subseteq \VEC^{\lW}_{cb}$.

As regards the first inclusion, let $L$ be a vector language of finite index (with or without erasing rules), and let $ind(L)=k$, $k\geq 1$. Then, there is a vector  grammar $G=(V, \Sigma, S, M)$ such that $L=L(G)$ and $ind(G)\leq k$. Without loss of generality we assume that $G$ is without repetitions. Let $R$ be the set of the rules of $M$.
By Theorem 16 in \cite{tur}, we can construct an $h$-Petri net controlled grammar $G'=(V, \Sigma, S, R, N_h)$, $N_h=(P\cup Q, T, F\cup E, \varphi, \zeta, \gamma, \mu_0, \tau)$, which is equivalent to the grammar $G$.
By definition, for every sentential form $w\in (V\cup\Sigma)^*$ in the grammar  $G$, $|w|_V\leq k$. It follows that $|w|_A\leq k$ for all $A\in V$. By bijection $\zeta:P\cup Q\to V\cup\{\lW\}$ we have $\mu(p)=\mu(\zeta^{-1}(A))\leq k$ for all $p\in P$ and $\mu \in \mathcal{R}(N_h, \mu_0)$, i.\,e., the corresponding cf Petri net $(P, T, F, \phi, \beta, \gamma, \iota)$ is with $k$-place capacity. Therefore~$G'$ is with weak place capacity.

On the other hand, the construction of an equivalent vector grammar for an $h$-Petri net controlled grammar,  can be extended to the case of weak capacities just by assigning the capacities of the corresponding places to the nonterminal symbols of the grammar.
\end{proof}

As regards strong capacities, there is no difference between weak and strong capacities for grammars controlled by $c$- and $s$-Petri nets because the number of tokens in every circle is limited by $1$. This yields:
\begin{corollary}\label{lem:wPNx=sPNx}
For $z\in \{c,s\}$, $\MAT_{\finite}=\mathbf{sPN}^{[\lW]}_{cz}$.
\end{corollary}

The only families not characterized yet are $\mathbf{sPN}^{[\lW]}_{ch}$. We conjecture that they are also equal to $\MAT_{\finite}$.



\section{Conclusions}
\label{sec:conclusions}

 We have introduced grammars with capacity bounds and their Petri net controlled counterparts. In particular, we have shown that their generative power lies strictly between the context-free languages of finite index and the matrix languages of finite index. Moreover, we studied extended context-free Petri nets with  place capacities. A possible extension of the concept is to use capacity functions that allow an unbounded number of some nonterminals.

The investigation shows that for every grammar controlled by a cf Petri net with $k$-place capacity, $k\geq 1$, there exists an equivalent grammar controlled by a cf Petri net with 1-place capacity, i.\,e., the families of languages generated by cf Petri nets with place capacities do not form a hierarchy with respect to the place capacities.

\bibliographystyle{eptcs}
\bibliography{stiebe}

\begin{thebibliography}{10}
\providecommand{\bibitemstart}[1]{\bibitem{#1}}
\providecommand{\bibitemend}{}
\providecommand{\bibliographystart}{}
\providecommand{\bibliographyend}{}
\providecommand{\url}[1]{\texttt{#1}}
\providecommand{\urlprefix}{Available at }
\providecommand{\bibinfo}[2]{#2}
\bibliographystart

\bibitemstart{cre:man}
\bibinfo{author}{S.~Crespi-Reghizzi} \& \bibinfo{author}{D.~Mandrioli}
  (\bibinfo{year}{1974}): \emph{\bibinfo{title}{Petri nets and commutative
  grammars}}.
\newblock \bibinfo{type}{Technical Report} \bibinfo{number}{74-5},
  \bibinfo{institution}{Laboraterio di Calcolatori, Instituto di Elettrotecnica
  ed Elettromca del Politecnico di Milano, Italy}.
\bibitemend

\bibitemstart{das:pau}
\bibinfo{author}{J.~Dassow} \& \bibinfo{author}{Gh. P\u{a}un}
  (\bibinfo{year}{1989}): \emph{\bibinfo{title}{Regulated Rewriting in Formal
  Language Theory}}.
\newblock \bibinfo{publisher}{Springer, Berlin}.
\bibitemend

\bibitemstart{das:tur}
\bibinfo{author}{J.~Dassow} \& \bibinfo{author}{S.~Turaev}
  (\bibinfo{year}{2008}): \emph{\bibinfo{title}{$k$-Petri net controlled
  grammars}}.
\newblock In: \bibinfo{editor}{C.~Mart\'{\i}n-Vide}, \bibinfo{editor}{F.~Otto}
  \& \bibinfo{editor}{H.~Fernau}, editors: {\sl \bibinfo{booktitle}{Language
  and Automata Theory and Applications. Second International Conference}}, {\sl
  \bibinfo{series}{LNCS}} \bibinfo{volume}{5196}.
  \bibinfo{publisher}{Springer}, pp. \bibinfo{pages}{209--220}.
\bibitemend

\bibitemstart{esp}
\bibinfo{author}{J.~Esparza} (\bibinfo{year}{1997}):
  \emph{\bibinfo{title}{Petri nets, commutative context-free grammars, and
  basic parallel processes}}.
\newblock {\sl \bibinfo{journal}{Fundam. Inf.}} \bibinfo{volume}{31}, pp.
  \bibinfo{pages}{13--25}.
\bibitemend

\bibitemstart{gin:spa1}
\bibinfo{author}{S.~Ginsburg} \& \bibinfo{author}{E.~H. Spanier}
  (\bibinfo{year}{1968}): \emph{\bibinfo{title}{Control sets on grammars}}.
\newblock {\sl \bibinfo{journal}{Math. Systems Theory}} \bibinfo{volume}{2},
  pp. \bibinfo{pages}{159--177}.
\bibitemend

\bibitemstart{gin:spa2}
\bibinfo{author}{S.~Ginsburg} \& \bibinfo{author}{E.~H. Spanier}
  (\bibinfo{year}{1968}): \emph{\bibinfo{title}{Derivation-bounded languages}}.
\newblock {\sl \bibinfo{journal}{J. Comput. Syst. Sci.}} \bibinfo{volume}{2},
  pp. \bibinfo{pages}{228--250}.
\bibitemend

\bibitemstart{hau:jan}
\bibinfo{author}{D.~Hauschildt} \& \bibinfo{author}{M.~Jantzen}
  (\bibinfo{year}{1994}): \emph{\bibinfo{title}{Petri net algorithms in the
  theory of matrix grammars}}.
\newblock {\sl \bibinfo{journal}{Acta Informatica}} \bibinfo{volume}{31}, pp.
  \bibinfo{pages}{719--728}.
\bibitemend

\bibitemstart{rei:roz}
\bibinfo{author}{W.~Reisig} \& \bibinfo{author}{G.~Rozenberg}
  (\bibinfo{year}{1998}): \emph{\bibinfo{title}{Lectures on Petri Nets I: Basic
  Models}}, {\sl \bibinfo{series}{LNCS}} \bibinfo{volume}{1491}.
\newblock \bibinfo{publisher}{Springer}.
\bibitemend

\bibitemstart{han}
\bibinfo{editor}{G.~Rozenberg} \& \bibinfo{editor}{A.~Salomaa}, editors
  (\bibinfo{year}{1997}): \emph{\bibinfo{title}{Handbook of Formal Languages,
  Volumes 1--3}}.
\newblock \bibinfo{publisher}{Springer-Verlag, Berlin}.
\bibitemend

\bibitemstart{tur}
\bibinfo{author}{S.~Turaev} (\bibinfo{year}{2007}): \emph{\bibinfo{title}{Petri
  net controlled grammars}}.
\newblock In: {\sl \bibinfo{booktitle}{Proc. 3rd Doctoral Workshop on
  MEMICS-2007, Znojmo, Czech Republic}}. pp. \bibinfo{pages}{233--240}.
\bibitemend

\bibliographyend
\end{thebibliography}

\end{document}